\newcount\Comments
\newcount\ResolvedComments
\documentclass{article}

\usepackage{fullpage}
\usepackage{natbib}
\newcommand{\ignore}[1]{}

\usepackage{amsmath,amsfonts,amssymb,bbm,amsthm} 
 \newtheorem{lemma}{Lemma}
 \newtheorem{theorem}{Theorem}
 \newtheorem{corollary}{Corollary}

 \newtheorem{definition}{Definition}

\allowdisplaybreaks[1]
\usepackage{microtype}

\usepackage[usenames,dvipsnames]{color}
\definecolor{darkgreen}{rgb}{0,0.5,0}
\newcommand{\kibitz}[2]{\ifnum\Comments=1\textcolor{#1}{#2}\fi}
\newcommand{\ag}[1]  {\kibitz{magenta}      {\bf\noindent [AG: #1]} }
\newcommand{\jrw}[1]  {\kibitz{red}      {\bf\noindent [JRW: #1]} }

\newcommand{\resolved}[1] {\ifnum\ResolvedComments=1\textcolor{blue}{[#1]}\fi}

\usepackage{dsfont}

\providecommand{\E}{\mathds{E}}\renewcommand{\E}{\mathds{E}}

\providecommand{\Ind}{\mathds{1}}\renewcommand{\Ind}{\mathds{1}}

\newcommand{\pr}{\textup{Pr}}

 \makeatletter
 \renewcommand{\paragraph}{%
   \@startsection{paragraph}{4}%
   {\z@}{0.5ex \@plus 1ex \@minus .2ex}{-1em}%
   {\normalfont\normalsize\bfseries}%
 }
 \makeatother

\newcommand{\union}{\cup}

\newcommand{\setdiff}{\backslash}

\begin{document}

\title{Incentivizing Evaluation via Limited Access to Ground Truth: \\Peer-Prediction Makes Things Worse}
\author{Xi Alice Gao \and James R. Wright \and Kevin Leyton-Brown}
\date{}
\maketitle

\begin{abstract}
In many settings, an effective way of evaluating objects of interest is to collect evaluations from dispersed individuals and to aggregate these evaluations together.  Some examples are categorizing online content and evaluating student assignments via peer grading. 
For this data science problem, one challenge is to motivate participants to conduct such evaluations carefully and to report them honestly, particularly when doing so is costly. Existing approaches, notably peer-prediction mechanisms, can incentivize truth telling in equilibrium.  However, they also give rise to equilibria in which agents do not pay the costs required to evaluate accurately, and hence fail to elicit useful information. We show that this problem is unavoidable whenever agents are able to coordinate using low-cost signals about the items being evaluated (e.g., text labels or pictures). We then consider ways of circumventing this problem by comparing agents' reports to ground truth, which is available in practice when there exist trusted evaluators---such as teaching assistants in the peer grading scenario---who can perform a limited number of unbiased (but noisy) evaluations. Of course, when such ground truth is available, a simpler approach is also possible: rewarding each agent based on agreement with ground truth with some probability, and unconditionally rewarding the agent otherwise. Surprisingly, we show that the simpler mechanism achieves stronger incentive guarantees given less access to ground truth than a large set of peer-prediction mechanisms.
\end{abstract}

\section{Introduction}

In many practical settings, an effective way of evaluating objects of interest is to collect evaluations from dispersed individuals and aggregate these evaluations together.  
For example, many millions of users rely on feedback from Rotten Tomatoes, Yelp and TripAdvisor to choose among competing movies, restaurants, and travel destinations.  Crowdsourcing platforms provide another example, enabling the collection of semantic labels of images and online content for use in training machine learning algorithms.  
This is a data science problem with two main challenges.  How should the collected data be aggregated to produce an accurate estimate?  How should incentives be designed to motivate participants to contribute high quality data?  In this paper, we focus on the incentive issues.

\resolved{
Crowdsourcing is a problem at the intersection between algorithmic game theory and data science,
in which evaluations are elicited from individuals and then aggregated.  Both information aggregation and incentive issues play a key role in this setting; in this paper we focus on incentive issues.
}


We are particularly motivated by the peer grading problem, which we will use as a running example.  Students benefit from open-ended assignments such as essays or proofs. However, such assignments are used relatively sparingly, particularly in large classes, because they require considerable time and effort to grade properly.  An efficient and scalable alternative is having students grade each other (and, in the process, learn from each other's work).  Many peer grading systems have been proposed and evaluated in the education literature~\citep{hamer2005method,cho2007scaffolded,pare2008peering,shah2013case,de2014crowdgrader,kulkarni2014scaling,raman2014methods,wright2015mechanical,caragiannis2015aggregating,de2015incentives}, albeit with a focus on evaluating the accuracy of grades collected under the assumption of full cooperation by students.
\resolved{Added 3 references to ordinal peer grading}

However, no experienced teacher would expect all students to behave nonstrategically when asked to invest effort in a time-consuming task. An effective peer grading system must therefore provide motivation for students to formulate evaluations carefully and to report them honestly. Many approaches have been developed to provide such motivation.  One notable category is peer-prediction methods~\citep{prelec2004bayesian,miller2005eliciting,jurca2009mechanisms,faltings2012eliciting,witkowski2012robust,witkowski2013dwelling,dasgupta2013crowdsourced,witkowski2013learning,radanovic2013robust,radanovic2014incentives,riley2014minimum,zhang2014elicitability,waggoner2014output,kamble2015truth,kong2016focal,shnayder2016informed}.  
%
In order to motivate each agent to reveal his private, informative signal, peer-prediction methods offer a reward based on how each agent's reports compare with those of his peers.  Such rewards are designed to induce truth telling in equilibrium---that is, they create a situation in which each agent has an interest in investing effort and revealing his private and informative signal truthfully, as long as he believes that all other agents will do the same. 

Even if they do offer a truthful equilibrium, peer-prediction methods also always induce other uninformative equilibria, 
the existence of which is inevitable \citep{jurca2009mechanisms,waggoner2014output}.  
Intuitively, if no other agent follows a strategy that depends on her private information, there is no reason for a given agent to deviate in a way that does so either: agents can only be rewarded for coordination, not for accuracy.  When private information is costly to obtain, uninformative equilibria are typically \emph{less} demanding for agents to play. This raises significant doubt about whether peer-prediction methods can motivate truthful reporting in practice. Experimental evaluations of peer-prediction methods have mixed results.  Some studies showed that agents reported truthfully \citep{shaw2011designing,john2012measuring,faltings2014incentives}; another study found that agents colluded on uninformative equilibria~\citep{gao2014trick}.

Recent progress on peer-prediction mechanisms has focused on making the truthful equilibrium Pareto dominant, i.e., (weakly) more rewarding to every agent than any other equilibrium~\citep{dasgupta2013crowdsourced,witkowski2013learning,kamble2015truth,radanovic2015incentives,shnayder2016informed}. 
This can be achieved by rewarding agents based on the distributions of their reports for multiple objects.  
However, we show in this paper that such arguments rely critically on the assumption that every agent has access to only one private signal per object. 
This is often untrue in practice; e.g., in peer grading, by taking a quick glance at an essay a student can observe characteristics such as length, formatting and the prevalence of grammatical errors.  These characteristics require hardly any effort to observe, can be arbitrarily uninformative about true quality, and are of no interest to the mechanism.  Yet their existence provides a means for the agents to coordinate. We build on this intuition to prove that no mechanism can guarantee that an equilibrium in which all agents truthfully report their informative signals is always Pareto dominant. Furthermore, we show that for any mechanism, the truthful equilibrium is \textit{always} Pareto dominated in some settings.

Motivated by these negative results, we move on to consider a setting in which the operator of the mechanism has access to trusted evaluators (e.g., teaching assistants) who can reliably provide noisy but informative signals of the object's true quality.  This allows for a hybrid mechanism that blends peer-prediction with comparison to trusted reports. With a fixed probability, the mechanism obtains a trusted report and rewards the agent based on the agreement between the agent's report and the trusted report \citep{Jurca:2005ue}.
Otherwise, the mechanism rewards the agent using a peer-prediction mechanism.  Such hybrid mechanisms can yield stronger incentive guarantees than other peer-prediction mechanisms, such as achieving truthful reporting of informative signals in Pareto-dominant equilibrium (see, e.g., \citep{Jurca:2005ue,dasgupta2013crowdsourced}).  Intuitively, if an agent seeks to be consistently close to a trusted report, then his best strategy is to reveal his informative signal truthfully. 

In fact, the availability of trusted reports is so powerful that it gives us the option of dispensing with peer-prediction altogether. Specifically, we can reward students based on agreement with the trusted report when the latter is available, but simply pay students a constant reward otherwise. Indeed, in \citet{wright2015mechanical} we introduced such a peer grading system and showed that it worked effectively in practice, based on a study across three years of a large class. This mechanism has even stronger incentive properties than the hybrid mechanism---because it induces a single-agent game, it can give rise to dominant-strategy truthfulness.

Our paper's main focus is on comparing these two approaches in terms of the number of trusted reports that they require. One might expect that the peer-prediction approach would have the edge, both because it relies on a weaker solution concept and because it leverages a second source of information reported by other agents. Surprisingly, we prove that this intuition is backwards.  We identify a simple sufficient condition, which, if satisfied, guarantees that the peer-insensitive mechanism offers the dominant strategy of truthful reporting of informative signals while querying trusted reports with a lower probability than is required for a peer-prediction mechanism to motivate truthful reporting in Pareto-dominant equilibrium.  We then show that all applicable peer-prediction mechanisms of which we are aware satisfy this sufficient condition.

\section{Peer-Prediction Mechanisms}

We begin by formally defining the game theoretic setting in which we will study the elicitation problem. 
%
A mechanism designer wishes to elicit information about a set $O$ of objects from $n$ risk-neutral agents. 
Each object $j$ has a latent quality $q_j \in Q$, where $Q$ is a finite set.  

\ag{New paragraph: Departure from standard assumption of one signal}
\jrw{lightly edited}
Agents have access to private information about the object of interest.
In the peer prediction literature, it is standard to assume that each agent receives information from a single, private signal.  Furthermore, this signal is assumed to be be the \emph{only} information that agent has about the object of interest.
However, we argue that, in reality, every agent can obtain multiple pieces of information with different quality by investing different amounts of efforts.
To capture this, we consider a simplified scenario by assuming that, for each object $j$, agent $i$ has access to two pieces of private information: a \emph{high-quality signal} $s^h_{ij} \in Q$ and a \emph{low-quality signal} $s^l_j$.

\ag{New paragraph: Description of high-quality signal}
The high-quality signal represents useful information about the object's quality that the mechanism designer wishes to elicit.  It is drawn from a distribution conditional on the object's actual quality $q_j$.  The joint distributions of the high-quality signals are common knowledge among the agents.  An agent $i$ can form a belief about the high-quality signal of another agent $i'$ by conditioning on his own high-quality signal.  Obtaining the high-quality signal requires a constant effort $c^E > 0$.

\ag{New paragraph: Description of low-quality signal}
The low-quality signal represents irrelevant information that the mechanism designer does not care about.   Yet it is easy to obtain and provides a way for agents to coordinate their reports.  For example, when evaluating essays, students can easily observe the number of grammatical mistakes or the apparent complexity of the language used without reading essays carefully.  Similarly, one could base a review on the decor without eating in a restaurant; evaluate the quality of a movie's trailer; etc.  For simplicity, we analyze the extreme case where the low-quality signal is uncorrelated with the object's true quality, is perfectly correlated across agents, and can be observed without effort.  Our results extend directly to a more general setting where agents can invest varying amounts of effort to obtain multiple signals with different degrees of correlation with the object's true quality.

\ag{Commented out a paragraph explaining how our model works for the peer grading setting}


Agents may strategize over both whether to incur the cost of effort to observe the high-quality signal and over what to report.
The mechanism designer's goal is to incentivize each agent to both observe the high-quality signal, and to truthfully report it.  We say that a mechanism has a \emph{truthful equilibrium} when it is an equilibrium for agents to observe the high-quality signal and truthfully report it (and, for some mechanisms, their posterior belief about other agents' high-quality signals).  

The mechanism designer's aim is to incentivize each agent $i \in \{1, \ldots, n\}$ to gather and truthfully report information about every object in $j \in O$.  Let $r_{ij}$ and $b_{ij}$ denote agent $i$'s signal and belief reports for object $j$ respectively.  A mechanism is defined by a reward function, which maps a profile of agent reports to a reward for each agent. We say that a mechanism is \emph{universal} if it can be applied without prior knowledge of the distribution from which signals are elicited, and for any number of agents greater than or equal to 3.

\begin{definition}[Universal peer-prediction mechanism]\label{def:universal}
A peer-prediction mechanism is \emph{universal} if it can be operated without knowledge of
the joint distribution of the high-quality signals $s^h_{ij}$ (i.e., it is ``detail free'' \citep{WilsonDoctrine87}) and well defined for any number of agents $n \geq 3$.
\end{definition}

We focus on universal mechanisms for two reasons.  First, in practice, it is extremely unrealistic to assume that a mechanism designer will have detailed knowledge of the joint signal distribution, so this allows us to focus on mechanisms that are more likely to be used in practice.  Second, it is relatively unrestrictive, as nearly all of the peer-prediction mechanisms in the literature satisfy universality.

Existing, universal peer-prediction mechanisms can be divided into three categories: output agreement mechanisms, multi-object mechanisms, and belief based mechanisms.

\paragraph{Output Agreement Mechanisms}

Output agreement mechanisms only collect signal reports from agents and reward an agent $i$ for evaluating object $j$ based on agents' signal reports for the object~\citep{faltings2012eliciting,witkowski2013dwelling,waggoner2014output}. 
\citet{waggoner2014output} and \citet{witkowski2013dwelling} studied the standard output agreement mechanism, where agent $i$ is only rewarded when his signal report matches that of another randomly chosen agent $j$.  Agent $i$'s reward is $z_i(r) = \Ind_{r_{ij} = r_{i'j}}$.
The~\citet{faltings2012eliciting} mechanism also rewards agents for agreement, scaled by the empirical frequency of the report agreed upon.  Agent $i$'s reward is $z_i(r) = \alpha + \beta \frac{\Ind_{r_{ij} = r_{i'j}}}{F(r_{i'j})}$, where $\alpha >0$ and $\beta >0$ are constants and $F(r_j)$ is the empirical frequency of $r_j$.

\paragraph{Multi-Object Mechanisms}

Multi-object mechanisms reward each agent based on his reports for multiple objects \citep{dasgupta2013crowdsourced,radanovic2015incentives,kamble2015truth,shnayder2016informed}.  
(The~\citet{shnayder2016informed} mechanism generalizes the~\citet{dasgupta2013crowdsourced} mechanism to the multi-signal setting.  Thus, we only refer to the~\citet{shnayder2016informed} mechanism below.)
  
The~\citet{shnayder2016informed} and~\citet{kamble2015truth} mechanisms also reward agents for agreement, as in output agreement mechanisms.  They extend output agreement mechanisms by adding additional scaling terms to the reward. These scaling terms are intended to exploit correlations between multiple tasks to make the truthful equilibrium dominate (a particular kind of) uninformative equilibria, by reducing the reward to agents who agree to an amount that is ``unsurprising'' given their reports on other objects.

The~\citet{shnayder2016informed} mechanism adds an additive scaling term to the reward for agreement.
To compute the scaling term, consider two sets of non-overlapping tasks $S_i$ and $S_{i'}$ such that agent $i$ has evaluated all objects in $S_i$ but none in $S_{i'}$ and agent $i'$ has evaluated all objects in $S_{i'}$ but none in $S_i$.  Let $F_{i}(s)$ and $F_{i'}(s)$ denote the frequency of signal $s \in Q$ in sets $S_i$ and $S_{i'}$ respectively.  Agent $i$ is rewarded according to $z_i(r) = \Ind_{r_{ij} = r_{i'j}} -  \sum_{s \in Q} F_i(s) F_{i'}(s)$.

In contrast, the~\citet{kamble2015truth} mechanism adds a multiplicative scaling term to the reward for agreement.
To compute the scaling term, choose 2 agents $k$ and $k'$ uniformly at random.  For each signal $s \in Q$, let $f^j(s) = \Ind_{r_{kj} = s} \Ind_{r_{k'j} = s}$.  Define $\hat{f}(s) = \sqrt{ \frac{1}{N} \sum_{j \in O} f^j(s) }$.  If $\hat{f}(s) \in \{0, 1\}$, then agent $i$'s reward is $0$.  Otherwise, agent $i$'s reward is $\Ind_{r_{ij} = r_{i'j}} \cdot  \frac{K }{\hat{f}(s)}$ for some constant $K > 0$.

The~\citet{radanovic2015incentives} mechanism rewards the agents for report agreement using a reward function inspired by the quadratic scoring rule.  To reward agent $i$ for evaluating object $j$, first choose another random agent $i'$ who also evaluated object $j$.  Then construct a sample $\Sigma_{i}$ of reports which contains one report for every object that is not evaluated by agent $i$.
The sample $\Sigma_{i}$ is double-mixed if it contains all possible signal realizations at least twice.
If $\Sigma_{i}$ is not double-mixed, agent $i$'s reward is $0$.  Otherwise, if $\Sigma_{i}$ is double-mixed, the mechanism chooses two objects $j'$ and $j''$ ($j' \ne j$, $j'' \ne j$ and $j' \ne j''$) such that the reports of $j'$ and $j''$ in the sample are the same as agent $i$'s report for $j$, i.e. $\Sigma_{i}(j') = \Sigma_{i}(j'') = r_{ij}$.  For each of $j'$ and $j''$, randomly select two reports $r_{i''j'}$ and $r_{i'''j''}$. Agent $i$'s is rewarded according to $z_i(r) = \frac{1}{2} + \Ind_{r_{i''j'} = r_{i'j}} - \frac{1}{2} \sum_{s \in Q} \Ind_{r_{i''j'} = s} \Ind_{r_{i'''j''} = s}.$

\paragraph{Belief Based Mechanisms}

Finally, some peer-prediction mechanisms collect both signal and belief reports from agents and reward each agent based on all agents' signal and belief reports for each object~\citep{witkowski2012robust,witkowski2013learning,radanovic2013robust,radanovic2014incentives,riley2014minimum}.  Below, let $R$ denote a proper scoring rule.

The robust Bayesian Truth Serum (BTS)~\citep{witkowski2012robust,witkowski2013learning} rewards agent $i$ for how well his belief report $b_i$ and shadowed belief report $b'_i$ predict the signal reports of another randomly chosen agent $k$.  Agent $i$'s reward is $z_i(r, b) = R(b'_i, r_k) + R(b_i, r_k)$.  Agent $i$'s shadowed belief report is calculated based on his signal report and another random agent $j$'s belief report: $b'_i = b_j + \delta$ if $r_i = 1$ and $b'_i = b_j - \delta$ if $r_i = 0$ where $\delta = \min(b_j,1 - b_j)$. 

The multi-valued robust BTS~\citep{radanovic2013robust} rewards agent $i$ if his signal report matches that of another random agent $j$ and his belief report accurately predicts agent $j$'s signal report. 
Agent $i$'s reward is $z_i(r, b) = \frac{1}{b_j(r_i)} \Ind_{r_i = r_j} + R(b_i, r_j)$.  

The divergence-based BTS~\citep{radanovic2014incentives} rewards agent $i$ if his belief report accurately predicts another random agent $j$'s signal report.  In addition, it penalizes agent $i$ if his signal report matches that of agent $j$ but his belief report is sufficiently different from that of agent $j$.
Agent $i$'s reward is $- \Ind_{r_i = r_j || D(b_i, b_j) > \theta} + R(b_i, r_j)$ where $D(||)$ is the divergence associated to the strictly proper scoring rule $R$, and $\theta$ is a parameter of the mechanism.

The~\citet{riley2014minimum} mechanism rewards agent $i$ for how well his belief report predicts other agents' signal reports.  Moreover, agent $i$'s reward is bounded above by the score for the average belief report of other agents reporting the same signal.  Formally, let $\delta_i = \min_{s \in Q} |\{r_j = s| j \ne i\}|$ be the minimum number of other agents who have reported any given signal.  If $\delta_i = 0$, agent $i$'s reward is $R(b_i, r_{-i})$.  
Otherwise, if $\delta_i \ge 1$, compute the proxy prediction $q_i(r_i)$ to be the average belief report for all other agents who made the same signal report as agent $i$.  Agent $i$'s reward is $\min \{ R(b_i, r_{-i}), R(q_i(r_i), r_{-i}) \}$.

\paragraph{Non-Universal Mechanisms}

We are aware of several additional peer-prediction mechanisms that we do not consider further in this paper because they are not universal in the sense of Definition~\ref{def:universal}.  
The~\citet{miller2005eliciting,zhang2014elicitability} and \citet{kong2016focal} mechanisms all
derive the agents' posterior beliefs based on their signal reports (hence requiring knowledge of the distribution from which signals are drawn); they all then reward the agents based on how well the derived posterior belief predicts other agents' signal reports using proper scoring rules.
The~\citet{jurca2009mechanisms} mechanism requires knowledge of the prior distribution over signals to construct rewards that either penalize or eliminate symmetric, uninformative equilibria. 
The Bayesian Truth Serum (BTS) mechanism~\citep{prelec2004bayesian} requires an infinite number of agents to guarantee the existence of the truthful equilibrium.  While we do not consider this mechanism, we note that \citet{prelec2004bayesian} pioneered the idea of eliciting both signal and belief reports from each agent.  This key idea was leveraged in much subsequent work to sustain the truthful equilibrium while not requiring knowledge of the prior distributions of the signals to operate the mechanism~\citep{witkowski2012robust,witkowski2013learning,radanovic2013robust,radanovic2014incentives,riley2014minimum}.

\paragraph{Hierarchical Mechanism~\citep{de2015incentives}}

\ag{Add citation to UC Santz Cruz's paper.  Should we revise definition of universality to make sure that this violates it?}
\jrw{No, I don't think so; I've edited to explicitly discuss the anonymity issue separately, see what you think.}
Independent to our work, \citet{de2015incentives} also proposed the idea of using peer prediction mechanisms in conjunction with limited access to trusted reports.  
In their hierarchical mechanism, students are placed into a tree structure.
Students in the top layer of the tree are incentivized through trusted reports whereas students in the layers below are incentivized via a peer prediction mechanism.  By an inductive argument, the truthful equilibrium exists and is unique, so long as the top-layer students are sufficiently incentivized.
This mechanism is detail free with respect to the distribution of signals, and is thus universal.
However, the existence of the truthful equilibrium requires every student to know which layer of the tree structure they occupy; that is, different students are treated differently ex-ante.
This is another example of work in which a widespread, seemingly innocuous assumption---in this case, anonymity; in the case of our own work, the single-signal assumption---turns out to have major implications.
In future work we intend to further explore relaxations of the single-signal assumption and anonymity, and connections between them.


\section{Impossibility of Pareto-Dominant, Truthful Elicitation}
\label{sec:impossibility}

In this section, 
we show that when agents have access to multiple signals about an object, Pareto-dominant truthful elicitation is impossible for any universal elicitation mechanism that computes agent rewards solely based on a profile of strategic agent reports (i.e., without any access to ground truth).
The intuition is that without knowledge of the distributions from which the signals are drawn, the mechanism cannot distinguish the signal that it hopes to elicit from other, irrelevant signals.  Thus, it cannot guarantee that the truthful equilibrium always yields the highest rewards to all agents. 

We focus on universal elicitation mechanisms that compute agent rewards solely based on a profile of agent reports.  Let $M$ denote such a mechanism. 
Let a \emph{signal structure} be a collection of signals $\{s_i\}_{i=1}^n$ drawn from a joint distribution $F$, where each agent $i$ observes $s_i$.
We say that a signal structure is \emph{$M$-elicitable} if there exists an equilibrium of $M$ where every agent $i$ truthfully reports $s_i$.  Let $\pi_i^F$ be agent $i$'s ex-ante expected reward in this equilibrium.
A \emph{multi-signal environment} is an environment in which the agents have access to at least two $M$-elicitable signal structures.  We refer to the signal structure that the mechanism seeks to elicit as the \emph{high-quality signal}, and all the others as \emph{low-quality signals}.  

\begin{theorem}
    For any universal elicitation mechanism, there exists a multi-signal environment in which the truthful equilibrium is not Pareto dominant. 
\begin{proof}
    Let $F,F'$ be $M$-elicitable signal structures such that $\pi^F_i \ge \pi^{F'}_i$ for all $i$, with $\pi^F_i > \pi^{F'}_i$ for some $i$.
    If no such pair of signal structures exists, then the result follows directly, since the truthful equilibrium does not Pareto dominate an equilibrium where agents report a low-quality signal.
    Otherwise, consider a multi-signal environment where the high-quality signal is distributed according to $F'$, and a low-quality signal is distributed according to $F$.
    The equilibrium in which agents reveal this low-quality signal Pareto dominates the truthful equilibrium in this environment.
\end{proof}
\end{theorem}

Now suppose that observing the high-quality signal is more costly to the agents than observing a low-quality signal.
Concretely, assume that observing the high-quality signal has an additive cost of $c_i>0$ for each agent $i$, and observing a low-quality signal has zero cost.  Call this a \emph{costly-observation multi-signal environment}.  In this realistic environment, an even stronger result holds.
\begin{theorem}
    For any universal elicitation mechanism, there exists a costly-observation multi-signal environment in which the truthful equilibrium is Pareto dominated.
    \begin{proof}
    Let $F,F'$ be $M$-elicitable signal structures such that $\pi^F_i \ge \pi^{F'}_i$ for all $i$.
    At least one such pair must exist, since every distribution has this relationship to itself.
    Fix a costly-observation multi-signal environment where the high-quality signal structure is jointly distributed according to $F'$, and a low-quality signal structure is jointly distributed according to $F$.
    Then each agent's expected utility in the truthful equilibrium is $\pi^{F'}_i - c_i < \pi^F_i$.  Hence every agent prefers the equilibrium in which agents reveal this low-quality signal, and the truthful equilibrium is Pareto dominated.
    \end{proof}
\end{theorem}
%
The essential insight of these results is that, in the presence of multiple elicitable signals, there is no way for a universal elicitation mechanism to be sure which signal it is eliciting.  
In particular, the truthful equilibrium is only Pareto dominant if the high-quality signal \emph{happens} to be drawn from a distribution yielding higher reward than \emph{every other} signal available to the agents.  In costly-observation environments, the element of luck is even stronger. The truthful equilibrium is Pareto dominant only if the high-quality signal structure happens to yield sufficiently high reward to compensate for the cost of observing the signals.

One way for the mechanism designer to ensure that agents are reporting the high-quality signal is to stochastically compare agents' reports to reports known to be correlated with that signal.  In the next section, we introduce a class of mechanisms that takes this approach.

\section{Combining Elicitation with Limited Access to Ground Truth}

Elicitation mechanisms are designed for situations where it is infeasible for the mechanism designer to evaluate each object herself.  However, in practice, it is virtually always possible, albeit costly, to obtain \emph{trusted reports}, i.e. unbiased evaluations of a subset of the objects.  In the peer grading setting, the instructor and teaching assistants can always mark some of the assignments. Similarly, review sites could in principle hire an expert to evaluate restaurants or hotels that its users have reviewed; and so on.

In this section, we define a class of mechanisms that take advantage of this limited access to ground truth to circumvent the result from Section~\ref{sec:impossibility}.
\begin{definition}[spot-checking mechanism]
  A \emph{spot-checking mechanism} is a tuple $M=(p,y,z)$, where $p$ is the \emph{spot check probability}; $y$ is a vector of functions $y_{ij}(r_{ij}, s^t_j)$ called the \emph{spot check mechanism}; and $z$ is a vector of functions $z_{ij}(b,r)$ called the \emph{unchecked mechanism}.

  Let $\Delta(Q)$ be the set of all distributions over the elements of $Q$.
  Each agent $i$ makes a \emph{signal report} $r_{ij} \in Q$, and a \emph{belief report} $b_{ij} \in \Delta(Q)$ for each object $j \in J_i$.
  The signal report is the signal that $i$ claims to have observed, and the belief report represents $i$'s posterior belief over the signal reports of the other agents.

  Agents may strategically choose whether or not to incur the cost of observing the high-quality signal, and having chosen which signal to observe, may report any function of either signal.
  Formally, let $G_i^h = \{g : Q \to Q\}$ be the set of all full-effort pure strategies, where an agent observes the high-quality signal---incurring observation cost $c^E$---and then reports a function $g(s^h_{ij})$ of the observed value.
  Let $D_l$ be the domain of $s_{ij}^l$.
  Let $G_i^l = \{g : D^l \to Q\}$ be the set of all no-effort pure strategies, where an agent observes the low-quality signal---incurring no observation cost---and then reports a function $g^l(s^l_{ij})$ of the observed value.
  The set of pure strategies available to an agent is thus $G_i^h \union G_i^l$.
  We assume that agents apply the same strategy to every object that they evaluate; however, we allow agents to play a mixed strategy by choosing the mapping stochastically.

  With probability $p$, the mechanism will \emph{spot check} an agent $i$'s report for a given object $j$.
  In this case, the mechanism obtains a \emph{trusted report}---that is, a sample from the signal $s^t_j$.
  The agent is then rewarded according to the spot check mechanism, applied to the profile of signal reports and spot checked objects.  With probability $1-p$, the object is not spot checked, and the agent is rewarded according to the unchecked mechanism.

  Thus, given a profile of signal reports $r \in \prod_{i \in N}Q^{J_i}$ and belief reports $b \in \prod_{i \in N} \Delta(Q)^{J_i}$, 
  an agent $i$ receives a reward of $\pi_i = \sum_{j \in J_i} \pi_{ij}$, where
  \begin{equation}
      \pi_{ij} = \begin{cases}
                     y_{ij}(r_i, s^t) &\text{if agent $i$'s report on object $j$ is spot checked,}\\
                     z_{ij}(b, r)         &\text{otherwise.}
                 \end{cases}
  \end{equation}
\end{definition}

We assume that the mechanism designer has no value for the reward given to the agents.  Instead, we seek only to minimize the probability of spot-checking required to make the truthful equilibrium either unique or Pareto dominant, since access to trusted reports is assumed to be costly.\footnote{If access to trusted reports were not costly, then querying strategic agents rather than trusted reports on all the objects would be pointless.}  This models situations where agents are rewarded by grades (as in peer grading), virtual points or badges (as in online reviews), or other artificial currencies.

The low-quality signal might be arbitrarily correlated with the underlying quality.
However, we assume that the high-quality signal is more correlated, in the sense that paying the cost of observing the high-quality signal is worthwhile.
Formally,
\[
     \E \left[ y_{ij}(s^h, s^t) - c^E \right] > \E \left[ y_{ij}(s^l, s^t) \right].
\]
That is, an agent who knows that they will be spot checked would prefer to pay the cost to observe the high-quality signal rather than observing the low-quality signal for free.
As an extreme example, if the low-quality signal were perfectly correlated with the quality, then no amount of spot-checking would induce an agent to observe the high-quality signal (nor, indeed, would a mechanism designer want them to).
\jrw{Might need more wordsmithing}

In this work we compare two approaches to using limited access to ground truth for elicitation.
The first approach is to augment existing peer-prediction mechanisms with spot-checking:
\begin{definition}[spot-checking peer-prediction mechanism]
    Let $z$ be a peer-prediction mechanism.
    Then any spot-checking mechanism that uses $z$ as its unchecked mechanism is a \emph{spot-checking peer-prediction mechanism}.
\end{definition}
The second approach is to rely exclusively on ground truth access to incentivize truthful reporting:

\begin{definition}[peer-insensitive mechanism]
    A \emph{peer-insensitive mechanism} is a spot-checking mechanism in which the unchecked mechanism is a constant function.  That is, $z_{ij}(b,r) = W$ for some constant $W > 0$.
\end{definition}

\section{When Does Peer-Prediction Help?}

We compare the peer-insensitive mechanism with all universal spot-checking peer-prediction mechanisms.  In Theorem~\ref{theorem_suff_condition_pareto}, we show that, if a simple sufficient condition is satisfied, then compared to all universal spot-checking peer-prediction mechanisms, the peer-insensitive mechanism can achieve stronger incentive properties (dominant-strategy truthfulness versus Pareto dominance of truthful equilibrium) while requiring a smaller spot check probability. 

We first define the $g^l$ strategy to be an agent's best no-effort strategy when a spot check is performed.
What is special about this strategy is that, if an agent chooses to invest no effort, then this is his best strategy for any spot check probability $p \in [0,1]$.  Thus, the $g^l$ equilibrium is stable and the best equilibrium for all agents conditional on not investing effort.
\begin{definition}
Let $g^l = \arg\max_{g \in G} \E[ y(g^l(s^l), s^t)]$ be an agent's best strategy when a spot check is performed and the agent invests no effort.
Let the $g^l$ equilibrium be the equilibrium where every agent uses the $g^l$ strategy. 
\end{definition}

In Lemma~\ref{lemma_prob_ds}, we analyze the peer-insensitive mechanism and derive an expression for the minimum spot check probability $p_{\textup{ds}}$ at which the truthful strategy is a dominant strategy for the peer-insensitive mechanism.  When the spot check probability is $p_{\textup{ds}}$, any agent is indifferent between playing the $g^l$ strategy and investing effort and reporting truthfully.

\newcommand{\labelDS}{\label{eqn_ds}}
\newcommand{\lemmaProbDS}{
The minimum spot check probability $p_{\textup{ds}}$ at which the truthful strategy is dominant for the peer-insensitive mechanism satisfies the following equation.
\begin{align}
p_{\textup{ds}}\, \E [ y(s^h, s^t) ] - c^E &= p_{\textup{ds}}\, \E[y(g^l(s^l), s^t)].  \labelDS
\end{align}
}

\begin{lemma} \label{lemma_prob_ds}
\lemmaProbDS
\end{lemma}
\renewcommand{\labelDS}{}

\begin{proof}
Please see Appendix \ref{sec:proof_lemma_ds}.
\end{proof}

Next, we consider any spot-checking peer-prediction mechanism.  Our goal is to derive a lower bound for $p_{\textup{Pareto}}$, the minimum spot check probability at which the truthful equilibrium is Pareto dominant.

For the truthful equilibrium to be Pareto dominant, it is necessary that the truthful equilibrium Pareto dominates the $g^l$ equilibrium.  This can be achieved in two ways.  If we can increase the spot check probability until $p_{\textup{el}}$ at which the $g^l$ equilibrium is eliminated, then the truthful equilibrium trivially Pareto dominates the $g^l$ equilibrium.  Otherwise, we can increase the spot check probability until $p_{\textup{ex}}$ at which the truthful equilibrium Pareto dominates the $g^l$ equilibrium assuming that the $g^l$ equilibrium exists when $p = p_{\textup{ex}}$.  Thus, $\min(p_{\textup{el}}, p_{\textup{ex}})$ is the minimum spot check probability at which the truthful equilibrium Pareto dominates the $g^l$ equilibrium, and it is also a lower bound for $p_{\textup{Pareto}}$.

In Lemma~\ref{lemma_el_ds}, we derive an expression for $p_{\textup{el}}$ and show that it is greater than or equal to $p_{\textup{ds}}$ under certain assumptions.  Intuitively, in order to eliminate the $g^l$ equilibrium, we need to increase the spot check probability enough such that an agent is persuaded to playing his best strategy with full effort rather than playing the $g^l$ strategy.  On one hand, the agent incurs a cost deviating from the $g^l$ equilibrium when all other agents follow it.  On the other hand, the agent's best strategy with full effort gives him no greater spot check reward than the truthful strategy.  The combined effect means that it is more costly to persuade an agent to deviate from the $g^l$ equilibrium than to motivate a single agent to report truthfully.

The sufficient conditions characterized in Lemmas~\ref{lemma_el_ds} and~\ref{lemma_ex_ds} and Theorem~\ref{theorem_suff_condition_pareto} are required to hold when $c^E=0$.  Note however, that if this condition is satisfied when $c^E=0$, then the consequents of these lemmas and theorems hold in settings with all positive cost of effort $c^E \ge 0$ as well.  Moreover, we will show that these sufficient conditions are satisfied by all universal peer-prediction mechanisms that we are aware of in the literature.

\newcommand{\lemmaElDS}{For any spot-checking peer-prediction mechanism, if the $g^l$ equilibrium exists when $c^E = 0$ and $p = 0$, then $p_{\textup{el}} \ge p_{\textup{ds}}$ for all $c^E \ge 0$.}
\begin{lemma}  \label{lemma_el_ds}
\lemmaElDS
\end{lemma}
\begin{proof}
	Please see Appendix \ref{sec:proof_lemma_el_ds}.
\end{proof}

In Lemma~\ref{lemma_ex_ds}, we show that $p_{\textup{ex}}$ is greater than or equal to $p_{\textup{ds}}$ under certain assumptions.  The intuition is that, when no spot check is performed, the $g^l$ equilibrium Pareto dominates the truthful equilibrium.  Thus, assuming that the $g^l$ equilibrium exists, it is more costly (in terms of increasing spot check probability) to make the truthful equilibrium Pareto dominate the $g^l$ equilibrium than to motivate a single agent to report truthfully.  

\newcommand{\lemmaExDS}{For any spot-checking peer-prediction mechanism, if the $g^l$ equilibrium exists and Pareto dominates the truthful equilibrium when $c^E = 0$ and $p=0$, then $p_{\textup{ex}} \ge p_{\textup{ds}}$ for all $c^E \ge 0$.}
\begin{lemma} \label{lemma_ex_ds}
\lemmaExDS
\end{lemma}
\begin{proof}
	Please see Appendix \ref{sec:proof_lemma_ex_ds}.
\end{proof}

If the conditions in Lemmas~\ref{lemma_el_ds} and~\ref{lemma_ex_ds} are satisfied, it is clear that $p_{\textup{Pareto}} \ge p_{\textup{ds}}$ because $\min(p_{\textup{el}}, p_{\textup{ex}})$, which lower bounds $p_{\textup{Pareto}}$, is already greater than or equal to $p_{\textup{ds}}$.  Thus, a sufficient condition for $p_{\textup{Pareto}} \ge p_{\textup{ds}}$ is simply all conditions in the two lemmas, as shown in Theorem~\ref{theorem_suff_condition_pareto}.

\newcommand{\theoremSuffConditionParetoI}{Sufficient condition for Pareto comparison}
\newcommand{\theoremSuffConditionParetoII}{For any spot-checking peer-prediction mechanism, if the $g^l$ equilibrium exists and Pareto dominates the truthful equilibrium when $c^E = 0$ and $p = 0$, then $p_{\textup{Pareto}} \ge p_{\textup{ds}}$ for all $c^E \ge 0$.}

\begin{theorem}[\theoremSuffConditionParetoI]
\label{theorem_suff_condition_pareto}
\theoremSuffConditionParetoII
\end{theorem}
\begin{proof}
	Please see Appendix \ref{sec:proof_theorem_suff_condition_pareto}.
\end{proof}

We now show that, under very natural conditions, \emph{every} universal peer-prediction mechanism of which we are aware in the literature satisfies the conditions of Theorem~\ref{theorem_suff_condition_pareto}; hence, in this setting, the peer-insensitive spot-checking mechanism requires less ground truth access than any spot-checking peer-prediction mechanism.

First, we assume that the low-quality signal $s^l$ is drawn from a uniform distribution over $Q$; this is essentially without loss of generality, since in any setting where the agents see a description of the object as well as their evaluation, a distribution of this form can be obtained by, e.g., hashing the description.  More realistically, objects may have names that are approximately uniformly distributed.
Second, we fix the spot check mechanism as in Equation~\eqref{eq:spot-check-fn}, using a form inspired by~\citet{dasgupta2013crowdsourced}.  Let $J^t$ be the set of objects that was spot-checked.
Let $i$ be an agent whose report $r_{ij}$ on object $j \in J_i$ has been spot checked.
Let $j' \in J_i$ be an object that $j$ evaluated, chosen uniformly at random, and let $j'' \in J^t \setdiff J_i$ be a spot-checked object, also chosen uniformly at random.\footnote{Note that in \citet{dasgupta2013crowdsourced}, it is important for strategic reasons that object $j'$ has not been evaluated by the opposing agent; this is not important in our setting, since the trusted reports are assumed to be nonstrategic.}  Then agent $i$'s reward for object $j$ is 
\begin{equation}
    y_{ij}(r_i, s^t) = \Ind_{r_{ij} = s^t_j} - \Ind_{r_{ij'} = s^t_{j''}}. \label{eq:spot-check-fn}
\end{equation}

\newcommand{\lemmaBestNoEffortStr}{For the spot check reward function in Equation~\eqref{eq:spot-check-fn}, an agent's best strategy conditional on not investing effort is always to report the low-quality signal $s^l$.}

\begin{lemma}
\label{lemma_best_no_effort_strategy}
	\lemmaBestNoEffortStr
\end{lemma}
\begin{proof}
	Please see Appendix~\ref{proof_lemma_best_no_effort_strategy}.
\end{proof}

\newcommand{\corollaryOne}{For spot-checking peer-prediction mechanisms based on~\citet{faltings2012eliciting,witkowski2013dwelling,dasgupta2013crowdsourced,waggoner2014output,kamble2015truth,radanovic2015incentives} and \citet{shnayder2016informed}, the minimum spot check probability $p_{\textup{Pareto}}$ for the Pareto dominance of the truthful equilibrium is greater than or equal to the minimum spot check probability $p_{\textup{ds}}$ at which the truthful strategy is a dominant strategy for the peer-insensitive mechanism.}

\begin{corollary} 
\label{corollary_1}
	\corollaryOne
\end{corollary}
\begin{proof}
	Please see Appendix \ref{proof_corollary_1}.
\end{proof}

\newcommand{\corollaryBeliefMech}{For spot-checking peer-prediction mechanisms based on ~\citet{witkowski2012robust,witkowski2013learning,radanovic2013robust,radanovic2014incentives} and \citet{riley2014minimum}, if the peer-prediction mechanism uses a symmetric proper scoring rule, then the minimum spot check probability $p_{\textup{Pareto}}$ for the Pareto dominance of the truthful equilibrium is greater than or equal to the minimum spot check probability $p_{\textup{ds}}$ at which the truthful strategy is a dominant strategy for the peer-insensitive mechanism.}

\begin{corollary}
\label{corollary_belief_mechanisms}
	\corollaryBeliefMech
\end{corollary}
\begin{proof}
	Please see Appendix \ref{proof_corollary_belief_mechanisms}.
\end{proof}

\section{Conclusions and Future Work}

We consider the problem of using limited access to noisy but unbiased ground truth to incentivize agents to invest costly effort in evaluating and truthfully reporting the quality of some object of interest.  Absent such spot-checking, peer-prediction mechanisms already guarantee the existence of a truthful equilibrium that induces both effort and honesty from the agents.  However, this truthful equilibrium may be less attractive to the agents than other, uninformative equilibria.

Some mechanisms in the literature have been carefully designed to ensure that the truthful equilibrium is the most attractive equilibrium to the agents (i.e., Pareto dominates all other equilibria).  However, these mechanisms rely crucially on the unrealistic assumption that agents' only means of correlating are via the signals that the mechanism aims to elicit.  We show that under  the more realistic assumption that agents have access to more than one signal, no universal peer-prediction mechanism has a Pareto dominant truthful equilibrium in all elicitable settings.

In contrast, we present a simpler peer-insensitive mechanism that provides incentives for effort and honesty only by checking the agents' reports against ground truth.  While one might have expected that peer-prediction would require less frequent access to ground truth to achieve stronger incentive properties than the peer-insensitive mechanism, we proved the opposite for all universal spot-checking peer-prediction mechanisms.

This surprising finding is intuitive in retrospect.  peer-prediction mechanisms can only motivate agents to behave in a certain way as a group.  An agent has a strong incentive to be truthful if all other agents are truthful; conversely, when all other agents coordinate on investing no effort, the agent again has a strong incentive to coordinate with the group.  peer-prediction mechanisms thus need to provide a strong enough incentive for agents to deviate from the most attractive uninformative equilibrium in the worst case, whereas the peer-insensitive mechanism only needs to motivate effort and honesty in an effectively single-agent setting.  

Many exciting future directions remain to be explored.
For example, we assumed that the principal does not care about the total amount of the artificial currency rewarded to the agents.  One possible direction would consider a setting in which the principal seeks to minimize both spot checks and the agents' rewards.
Also, in our analysis, we assumed that the spot check probability does not depend on the agents' reports.  Conditioning the spot check probability on the agents' reports might allow the mechanism to more efficiently detect and punish uninformative equilibria.

\bibliographystyle{plainnat}
\bibliography{references,peerprediction,peergrading}

\appendix

\providecommand{\restatetheorem}[1]{\renewcommand{\thetheorem}{\ref{#1}}}
\providecommand{\restatelemma}[1]{\renewcommand{\thetheorem}{\ref{#1}}}
\providecommand{\restatecorollary}[1]{\renewcommand{\thetheorem}{\ref{#1}}}

\section{Proof of Lemma~\ref{lemma_prob_ds}}
\label{sec:proof_lemma_ds}

\restatelemma{lemma_prob_ds}
\begin{lemma}
	\lemmaProbDS
\end{lemma}
\begin{proof}
Consider the peer insensitive mechanism with a fixed spot check probability $p \ge 0$.  
When an agent uses the truthful strategy, his expected utility is
\begin{align}
p\, \E [y(s^h, s^t) ] + (1-p)\, W - c^E. \label{eqn_ds_1}
\end{align}
When an agent invests no effort, his best strategy is $g^l$.  His expected utility from playing the $g^l$ strategy is
\begin{align}
p\, \E[y(g^l(s^l), s^t) ]  + (1-p)\, W. \label{eqn_ds_2}
\end{align}

When $p = p_{\textup{ds}}$, it must be that an agent's expected utilities in the above two expressions~\eqref{eqn_ds_1} and~\eqref{eqn_ds_2} are the same.  
\begin{align}
p_{\textup{ds}}\, \E [y(s^h, s^t) ] + (1-p_{\textup{ds}})\, W - c^E &= p_{\textup{ds}}\, \E[y(g^l(s^l), s^t) ]  + (1-p_{\textup{ds}})\, W  \notag \\
p_{\textup{ds}}\, \E [y(s^h, s^t) ] - c^E &= p_{\textup{ds}}\, \E[y(g^l(s^l), s^t) ]. \notag
\end{align}
\end{proof}

\section{Proof of Lemma~\ref{lemma_el_ds}}
\label{sec:proof_lemma_el_ds}

\restatelemma{lemma_el_ds}
\begin{lemma}
	\lemmaElDS
\end{lemma}

\begin{proof}
Recall that $p_{\textup{el}}$ is the minimum spot check probability at which the $g^l$ equilibrium is eliminated.  
We first derive an expression for $p_{\textup{el}}$.

We consider a spot checking peer prediction mechanism.  By our assumption, the $g^l$ equilibrium exists when $c^E = 0$ and the spot check probability is $0$.  

Assume that all other agents play the $g^l$ strategy and analyze agent $i$'s best response.
First, we note that, if agent $i$ invests no effort, then agent $i$'s best strategy is the $g^l$ strategy for any spot check probability.  
(To maximize his spot check reward $y$, he should play the $g^l$ strategy by the definition of the $g^l$ strategy.  To maximize his non spot check reward, his best strategy is also the $g^l$ strategy because the $g^l$ equilibrium exists at $p = 0$.  )
Thus, to eliminate the $g^l$ equilibrium, we need to increase the spot check probability until agent $i$ prefers to play his best strategy conditional on investing full effort.  

Consider a fixed spot check probability $p$ and suppose that the $g^l$ equilibrium exists at this spot check probability.  Suppose that all other agents play the $g^l$ strategy.  

If agent $i$ does not invest effort, his best response is to also play the $g^l$ strategy and his expected utility is
\begin{align}
p\, \E[y(g^l(s^l), s^t) ] + (1-p)\, \E[z(g^l(s^l), g^l(s^l))]. \label{eqn_el_1}
\end{align}

If agent $i$ invests full effort, let $g^{\textup{br}}$ denote agent $i$'s best response and his expected utility by playing this best response is
\begin{align}
p\, \E[ y(g^{\textup{br}}(s^h), s^t)  ] + (1-p)\, \E[z(g^{\textup{br}}(s^h), g^l(s^l))] - c^E. \label{eqn_el_2}
\end{align}

By definition of $p_{\textup{el}}$, when $p = p_{\textup{el}}$, an agent's expected utility in the above two expressions~\eqref{eqn_el_1} and~\eqref{eqn_el_2} are the same.  Thus $p_{\textup{el}}$ must satisfy
\begin{align}
&p_{\textup{el}}\, \E [ y(g^{\textup{br}}(s^h), s^t)  ] + (1-p_{\textup{el}})\, \E[ z(g^{\textup{br}}(s^h), g^l(s^l))  ] - c^E \notag\\
&\qquad\qquad = p_{\textup{el}}\, \E[ y(g^l(s^l), s^t)  ] + (1-p_{\textup{el}})\, \E[ z(g^l(s^l), g^l(s^l))  ] \notag \\
&p_{\textup{el}}\, \E [ y(g^{\textup{br}}(s^h), s^t)  ] + (1-p_{\textup{el}})\, (\E[ z(g^{\textup{br}}(s^h), g^l(s^l))  ] - \E[ z(g^l(s^l), g^l(s^l))  ]) - c^E \notag\\
&\qquad\qquad = p_{\textup{el}}\, \E[ y(g^l(s^l), s^t)  ]. \label{unique_pprime}
\end{align}

Next, we would like to show that $p_{\textup{el}} \ge p_{\textup{ds}}$.

Since the $g^l$ equilibrium exists when $c^E = 0$ and $p = 0$, it follows from the definition of equilibrium that
\begin{equation}
    \E[ z(g^{\textup{br}}(s^h), g^l(s^l))  ] \le \E[ z(g^l(s^l), g^l(s^l)) ]. \label{eq:existence-comp}
\end{equation}

Taking $p_{\textup{el}}$ and substituting into the LHS of~\eqref{eqn_ds} (definition of $p_{\textup{ds}}$), in a setting with arbitrary positive $c^E \ge 0$, we have
\begin{align}
& p_{\textup{el}}\, \E [ y(s^h, s^t) ] - c^E \notag\\
&\ge  p_{\textup{el}}\, \E [ y(s^h, s^t) ] + (1-p_{\textup{el}})\, (\E[ z(g^{\textup{br}}(s^h), g^l(s^l))  ] - \E[ z(g^l(s^l), g^l(s^l))  ]) - c^E \label{unique_eqn1} \\
&> p_{\textup{el}}\, \E [ y(g^{\textup{br}}(s^h), s^t)  ]  + (1-p_{\textup{el}})\, (\E[ z(g^{\textup{br}}(s^h), g^l(s^l))  ] - \E[z(g^l(s^l), g^l(s^l))  ]) - c^E \label{unique_eqn2} \\
&=  p_{\textup{el}}\, \E[y(g^l(s^l), s^t)]. \label{unique_eqn3}
\end{align}
Inequality~\eqref{unique_eqn1} holds due to Equation~\eqref{eq:existence-comp}.
Inequality~\eqref{unique_eqn2} holds due to the truthfulness of spot checks: reporting high-quality signal maximizes the spot check reward.  Equation~\eqref{unique_eqn3} follows from Equation~\eqref{unique_pprime}.

Thus, if we substitute $p_{\textup{el}}$ into Equation~\eqref{eqn_ds}, then the resulting LHS is greater than the RHS.  By definition of $p_{\textup{ds}}$, it is the minimum spot check probability for which the LHS of~\eqref{eqn_ds} is     greater than its RHS.  Thus, it must be that $p_{\textup{el}} \ge p_{\textup{ds}}$.
\end{proof}

\section{Proof of Lemma~\ref{lemma_ex_ds}}
\label{sec:proof_lemma_ex_ds}

\restatelemma{lemma_ex_ds}
\begin{lemma}
	\lemmaExDS
\end{lemma}
\begin{proof}
Recall that $p_{\textup{ex}}$ is the minimum spot check probability at which the $g^l$ equilibrium Pareto dominates the truthful equilibrium while the $g^l$ equilibrium exists at $p = p_{\textup{ex}}$.  
We first derive an expression for $p_{\textup{ex}}$.

We consider a spot checking peer prediction mechanism.  By our assumption, the $g^l$ equilibrium exists and Pareto dominates the truthful equilibrium when $c^E = 0$ and $p=0$.

Consider a fixed spot check probability $p \ge 0$.  Assume that the $g^l$ equilibrium exists at this spot check probability.
At the truthful equilibrium, an agent's expected utility is 
\begin{align}
p\, \E [y(s^h, s^t)]  + (1-p)\, \E [z(s^h, s^h)] - c^E. \label{eqn_ex_1}
\end{align}
At the $g^l$ equilibrium, an agent's expected utility is
\begin{align}
p\, \E [y(g^l(s^l), s^t)] + (1-p)\, \E [z(g^l(s^l), g^l(s^l))]. \label{eqn_ex_2}
\end{align}

When $p = p_{\textup{ex}}$, it must be that an agent's expected utility in the above two expressions~\eqref{eqn_ex_1} and~\eqref{eqn_ex_2} are the same.  Thus $p_{\textup{ex}}$ must satisfy
\begin{align}
&p_{\textup{ex}}\, \E [y(s^h, s^t)]  + (1- p_{\textup{ex}})\, \E [z(s^h, s^h)] - c^E \notag\\
&\qquad\qquad = p_{\textup{ex}}\, \E [y(g^l(s^l), s^t)] + (1-p_{\textup{ex}})\, \E [z(g^l(s^l), g^l(s^l))] \notag\\
&p_{\textup{ex}}\, \E [y(s^h, s^t)]  + (1- p_{\textup{ex}})\, \left( \E [z(s^h, s^h)] - \E [z(g^l(s^l), g^l(s^l))] \right) - c^E \notag\\
&\qquad\qquad = p_{\textup{ex}}\, \E [y(g^l(s^l), s^t)]. \label{eqn_pareto_1}
\end{align}

Next, we would like to show that $p_{\textup{ex}} \ge p_{\textup{ds}}$.

Since the $g^l$ equilibrium exists and Pareto dominates the truthful equilibrium for $c^E=0$ and $p=0$, it follows from the definition of Pareto dominance that
\begin{equation}
    \E [z(s^h, s^h)] \le \E [z(g^l(s^l), g^l(s^l))]. \label{eq:pareto-comp}
\end{equation}

Taking $p_{\textup{ex}}$ and substituting it into the LHS of Equation~\eqref{eqn_ds} (definition of $p_{\textup{ds}}$), in a setting with arbitrary positive $c^E \ge 0$, we have
\begin{align}
&p_{\textup{ex}}\, \E [y(s^h, s^t) ] - c^E \notag\\
&\ge p_{\textup{ex}}\, \E [y(s^h, s^t)]  + (1- p_{\textup{ex}})\, \left( \E [z(s^h, s^h)] - \E [z(g^l(s^l), g^l(s^l))] \right) - c^E \label{eqn_pareto_2}\\
&=p_{\textup{ex}}\, \E [y(g^l(s^l), s^t)]  \label{eqn_pareto_3}
\end{align}
Equation~\eqref{eqn_pareto_2} follows from Equation~\eqref{eq:pareto-comp}.
Equation~\eqref{eqn_pareto_3} follows from Equation~\eqref{eqn_pareto_1}. 

Thus, if we substitute $p_{\textup{ex}}$ into Equation~\eqref{eqn_ds}, then the resulting LHS is weakly greater than the RHS.  By definition of $p_{\textup{ds}}$, it is the minimum spot check probability for which the LHS of~\eqref{eqn_ds} is greater than its RHS.  Thus, it must be that $p_{\textup{ex}} \ge p_{\textup{ds}}$.
\end{proof}

\section{Proof of Theorem~\ref{theorem_suff_condition_pareto}}
\label{sec:proof_theorem_suff_condition_pareto}

\restatetheorem{theorem_suff_condition_pareto}
\begin{theorem}[\theoremSuffConditionParetoI]
	\theoremSuffConditionParetoII
\end{theorem}
\begin{proof}
Consider any spot checking peer prediction mechanism.  

For the truthful equilibrium to be Pareto dominant, it is necessary that either the $g^l$ equilibrium is eliminated or the truthful equilibrium Pareto dominates the $g^l$ equilibrium while the $g^l$ equilibrium exists.
$p_{\textup{el}}$ is the minimum spot check probability at which the $g^l$ equilibrium is eliminated. 
$p_{\textup{ex}}$ is the minimum spot check probability at which the truthful equilibrium Pareto dominates the $g^l$ equilibrium while the $g^l$ equilibrium exists at $p = p_{\textup{ex}}$.  Thus, the minimum of $p_{\textup{el}}$ and $p_{\textup{ex}}$ is a lower bound of $p_{\textup{pareto}}$.  Formally
\begin{align}
p_{\textup{pareto}} \ge \min(p_{\textup{el}}, p_{\textup{ex}}).  \label{eqn_suff_1}
\end{align}  
By assumption, the $g^l$ equilibrium exists when $p = 0$.  By Lemma~\ref{lemma_el_ds}, we have
\begin{align}
p_{\textup{el}} \ge p_{\textup{ds}}.  \label{eqn_suff_2}
\end{align}
By assumption, the $g^l$ equilibrium exists and Pareto dominates the truthful equilibrium when $p = 0$.  By Lemma~\ref{lemma_ex_ds}, we have
\begin{align}
p_{\textup{ex}} \ge p_{\textup{ds}}. \label{eqn_suff_3}
\end{align}

By Equations~\eqref{eqn_suff_1},~\eqref{eqn_suff_2} and~\eqref{eqn_suff_3}, we have
\begin{align}
p_{\textup{pareto}} &\ge \min(p_{\textup{el}}, p_{\textup{ex}}) \notag\\
&\ge \min(p_{\textup{ds}}, p_{\textup{ex}}) \notag\\
&\ge \min(p_{\textup{ds}}, p_{\textup{ds}}) \notag\\
&= p_{\textup{ds}}. \notag
\end{align}
\end{proof}

\section{Proof of Lemma~\ref{lemma_best_no_effort_strategy}}
\label{proof_lemma_best_no_effort_strategy}

\restatelemma{lemma_best_no_effort_strategy}
\begin{lemma}
	\lemmaBestNoEffortStr
\end{lemma}

\begin{proof}
Consider the spot check reward mechanism in Equation~\eqref{eq:spot-check-fn}.  

If an agent invests no effort, his expected spot check reward is:
\begin{align}
&\sum_{s \in Q} \pr(r = s) \left( \pr(s^t = s | r = s) - \sum_{s' \in Q} \pr(s^t = s') \pr(r = s') \right) \notag\\
&= \sum_{s \in Q} \pr(s^t = s, r = s) - \sum_{s' \in Q} \pr(s^t = s') \pr(r = s') \notag
\end{align}

If the agent always makes a fixed report $r$, then the TA's signal $s^t$ and the agent's report $r$ are independent random variables, i.e.
\[\pr(s^t = s, r = s) = \pr(s^t = s) \pr(r = s),\]
for any $s \in Q$.  Thus the agent's expected reward must be zero.
\begin{align}
&\sum_{s \in Q} \pr(s^t = s, r = s) - \sum_{s' \in Q} \pr(s^t = s') \pr(r = s') \notag\\
&= \sum_{s \in Q} \pr(s^t = s) \pr(r = s) - \sum_{s' \in Q} \pr(s^t = s') \pr(r = s') \notag\\ 
&= 0 \notag
\end{align}

If the agent truthfully reports the low-quality signal $s^l$, then the agent's expected reward is:
\begin{align}
&\sum_{s \in Q} \pr(r = s) \left( \pr(s^t = s | r = s) - \sum_{s' \in Q} \pr(s^t = s') \pr(r = s') \right) \notag\\
&= \sum_{s \in Q} \pr(r = s) \left( \pr(s^t = s | r = s) - \pr(s^t = s') \right) \notag\\
&\ge 0 \notag
\end{align}

Thus the agent's expected spot check reward is maximized when he reports the low-quality signal $s^l$.
\end{proof}

\section{Proof of Corollary~\ref{corollary_1}}
\label{proof_corollary_1}

\restatecorollary{corollary_1}
\begin{corollary}
	\corollaryOne
\end{corollary}

\begin{proof}
By Lemma~\ref{lemma_best_no_effort_strategy}, for any spot checking peer prediction mechanism, the $g^l$ strategy is to always report the low-quality signal $s^l$.  

To verify that the conditions of Theorem~\ref{theorem_suff_condition_pareto} are satisfied, it suffices to verify that when $p = 0$, the $s^l$ equilibrium of the peer prediction mechanism exists and Pareto dominates the truthful equilibrium.  We verify these two conditions for all of the listed peer prediction mechanisms below.

We first consider output agreement peer prediction mechanisms. 
\paragraph{The Standard Output Agreement Mechanism~\citep{witkowski2013dwelling,waggoner2014output}} 

When $c^E = 0$ and $p = 0$, the $s^l$ equilibrium exists. (If all other agents except $i$ report $s^l$, then agent $i$'s best response is to also report $s^l$ in order to perfectly agree with other reports.)

When $c^E = 0$ and $p = 0$, at the $s^l$ equilibrium, every agent's expected utility is $1$ because their reports always perfectly agree.

When $c^E = 0$ and $p = 0$, at the truthful equilibrium, an agent's expected utility is 
$$\sum_{s^h \in Q} \pr(s^h) \pr(s^h|s^h) < \sum_{s^h \in Q} \pr(s^h) = 1,$$
where the inequality is due to the fact that the high-quality signals are noisy.  That is, for every realization $s^h$ of the high-quality signal, $\pr(s^h|s^h) \le 1$ and there exists one realization $s^h$ of the high-quality signal such that $\pr(s^h|s^h) < 1$.  
Thus, the $s^l$ equilibrium Pareto dominates the truthful equilibrium when $c^E = 0$ and $p = 0$.
The conditions of Theorem~\ref{theorem_suff_condition_pareto} are therefore satisfied, and hence $p_{\textup{Pareto}} \ge p_{\textup{ds}}$ for all settings with positive effort cost $c^E \ge 0$.

\paragraph{Peer Truth Serum~\citep{faltings2012eliciting}} 

When $c^E = 0$ and $p = 0$, the $s^l$ equilibrium exists.  (If all other agents except $i$ report $s^l$, then agent $i$'s best response is to also report $s^l$.)

When $c^E = 0$ and $p = 0$, at the $s^l$ equilibrium, everyone reports $s^l$ and the empirical frequency of $s^l$ reports is 1 ($F(s^l) = 1$).  Thus, every agent's expected utility is 
$$\alpha + \beta \frac{1}{F(s^l)} = \alpha + \beta.$$
When $c^E = 0$ and $p = 0$, at the truthful equilibrium, if agent receives the high-quality signal $s^h$ for an object, then he expects the empirical frequency of this signal to be $\pr(s^h | s^h)$.  Thus, at this equilibrium, an agent's expected utility is 
$$\displaystyle \alpha + \beta \sum_{s^h \in Q} \pr(s^h) \pr(s^h | s^h) \frac{1}{\pr(s^h | s^h)} = \alpha + \beta.$$
Thus, the $s^l$ equilibrium (weakly) Pareto dominates the truthful equilibrium when $c^E = 0$ and $p = 0$.
The conditions of Theorem~\ref{theorem_suff_condition_pareto} are therefore satisfied, and hence $p_{\textup{Pareto}} \ge p_{\textup{ds}}$ for all settings with positive effort cost $c^E \ge 0$.

Next, we consider multi-object peer prediction mechanisms.

\paragraph{\citet{dasgupta2013crowdsourced,shnayder2016informed}} 

When $c^E = 0$ and $p = 0$, the $s^l$ equilibrium exists.
(If all other agents always report the low-quality signal $s^l$ for every object, then agent $i$'s best response is also to report $s^l$ in order to maximize the probability of his report agreeing with other agents' reports for the same object.)
 
When $p = 0$, at the $s^l$ equilibrium, an agent's expected utility is 
\begin{align}
&\sum_{s^l \in Q} \pr(s^l) \pr(s^l | s^l) - \sum_{s^l \in Q} \pr(s^l) \pr(s^l) 
= \sum_{s^l \in Q} \pr(s^l) - \sum_{s^l \in Q} \pr(s^l) \pr(s^l) \notag\\
&= 1 - \sum_{s^l \in Q} \frac{1}{|Q|^2} 
= 1 - \frac{1}{|Q|}, \notag
\end{align}
where the first equality was due to the fact that the low-quality signal $s^l$ is noiseless ($\pr(s^l | s^l) = 1$) and the second equality was due to the fact that $s^l$ is drawn from a uniform distribution ($\pr(s^l) = \frac{1}{|Q|}$).

When $c^E = 0$ and $p = 0$, at the truthful equilibrium, an agent's expected utility is 
\begin{align}
&\sum_{s^h \in Q} \pr(s^h) \pr(s^h | s^h) - \sum_{s^h \in Q} \pr(s^h) \pr(s^h)
< \sum_{s^h \in Q} \pr(s^h) - \sum_{s^h \in Q} \pr(s^h)^2 \notag\\
&= 1 - \sum_{s^h \in Q} \pr(s^h)^2
\le 1 - \frac{1}{|Q|}, \notag
\end{align}
where the first inequality was due to the fact that the high-quality signal is noisy.  That is, for every realization $s^h$ of the high-quality signal, $\pr(s^h|s^h) \le 1$ and there exists one realization $s^h$ of the high-quality signal such that $\pr(s^h|s^h) < 1$.  Thus, the $s^l$ equilibrium Pareto dominates the truthful equilibrium when $c^E = 0$ and $p = 0$.
The conditions of Theorem~\ref{theorem_suff_condition_pareto} are therefore satisfied, and hence $p_{\textup{Pareto}} \ge p_{\textup{ds}}$ for all settings with positive effort cost $c^E \ge 0$.

\paragraph{\citet{kamble2015truth}} 

When $c^E = 0$ and $p = 0$, the $s^l$ equilibrium exists.  (If all other agents always report $s^l$, an agent's best response is also to report $s^l$ because doing so maximizes the probability of his report agreeing with other agents' reports for the same object.)

When $c^E = 0$ and $p = 0$, at the $s^l$ equilibrium, an agent's expected utility is 
\begin{align}
&\displaystyle \sum_{s^l \in Q} \pr(s^l) \pr(s^l | s^l) \lim_{N \rightarrow \infty} r(s^l) 
= \sum_{s^l \in Q} \pr(s^l) \frac{K}{ \sqrt{ \pr(s^l, s^l) } } 
= K \sum_{s^l \in Q}  \frac{ \pr(s^l) }{ \sqrt{ \pr(s^l) } } \notag\\
&= K \sum_{s^l \in Q}  \sqrt{ \pr(s^l) } 
= K \sum_{s_l \in Q} \sqrt{ \frac{1}{|Q|} }, \notag
\end{align}
where the first two equalities were due to the fact that the low-quality signal $s^l$ is noiseless ($\pr(s^l | s^l) = \pr(s^l)$), and the final equality was due to the fact that the low-quality signal $s^l $ is drawn from a uniform distribution.

When $c^E = 0$ and $p = 0$, at the truthful equilibrium, an agent's expected utility is 
\begin{align}
&\displaystyle \sum_{s^h \in Q} \pr(s^h) \pr(s^h | s^h) \lim_{N \rightarrow \infty} r(s^h) 
= \sum_{s^h \in Q} \pr(s^h, s^h) \frac{K}{ \sqrt{ \pr(s^h, s^h) } }  \notag\\
&= K \sum_{s^h \in Q} \sqrt{ \pr(s^h, s^h) } 
< K \sum_{s^h \in Q} \sqrt{ \pr(s^h) }
\le K \sum_{s^h \in Q} \sqrt{\frac{1}{ |Q| }}, \notag
\end{align}
where the first inequality was due to the fact that the high-quality signal $s^h$ is noisy.  That is, for every realization $s^h$ of the high-quality signal, $\pr(s^h|s^h) \le 1$ and there exists one realization $s^h$ of the high-quality signal such that $\pr(s^h|s^h) < 1$.
Thus, the $s^l$ equilibrium Pareto dominates the truthful equilibrium when $c^E = 0$ and $p = 0$.
The conditions of Theorem~\ref{theorem_suff_condition_pareto} are therefore satisfied, and hence $p_{\textup{Pareto}} \ge p_{\textup{ds}}$ for all settings with positive effort cost $c^E \ge 0$.

\paragraph{\citet{radanovic2015incentives}} 

When $c^E = 0$ and $p = 0$, the $s^l$ equilibrium exists.  
(If all other agents always report $s^l$ for every object, then any sample taken will not be ``double mixed''.\footnote{A sample is double mixed if every possible value appears at least twice.  This mechanism behaves differently depending on whether or not it collects a double mixed sample of reports from the agents.}  Thus, an agent's expected utility is zero regardless of his strategy.  In particular also reporting $s^l$ for every object is a best response.)

When $c^E = 0$ and $p = 0$, at the $s^l$ equilibrium, it must be that $r_{i''j'} = r_{i'j}$ and $r_{i''j'} = r_{i'''j''} = r_{ij}$.  An agent's expected utility at the $s^l$ equilibrium is:
\begin{align}
&\frac{1}{2} + \Ind_{r_{i''j'} = r_{i'j}} - \frac{1}{2} \sum_{s \in Q} \Ind_{r_{i''j'} = s} \Ind_{r_{i'''j''} = s} 
= \frac{1}{2} + 1 - \frac{1}{2} * 1 = 1. \notag
\end{align}

Let $\pi(\Sigma)$ be the probability that the sample $\Sigma$ is double mixed.  
When $c^E = 0$ and $p = 0$, at the truthful equilibrium, an agent's expected utility is:
\begin{align}
&\pi(\Sigma) \left( \frac{1}{2} + \pr(r_{i''j'} | r_{ij}) - \frac{1}{2} \sum_{s \in Q} \pr(s | r_{ij})^2 \right) 
\le \frac{1}{2} + \pr(r_{i''j'} | r_{ij}) - \frac{1}{2} \sum_{s \in Q} \pr(s | r_{ij})^2 \notag\\
&\le \frac{1}{2} + 1 - \frac{1}{2} * 1 = 1, \notag
\end{align}
where the first inequality is due to the fact that $\pi(\Sigma) \le 1$ and the second inequality was due to the fact that the agent's expected utility is maximized when $\pr(r_{i''j'} | r_{ij}) = 1$.  
Thus, the $s^l$ equilibrium Pareto dominates the truthful equilibrium when $c^E = 0$ and $p = 0$.
The conditions of Theorem~\ref{theorem_suff_condition_pareto} are therefore satisfied, and hence $p_{\textup{Pareto}} \ge p_{\textup{ds}}$ for all settings with positive effort cost $c^E \ge 0$.

\end{proof}

\section{Proof of Corollary~\ref{corollary_belief_mechanisms}}
\label{proof_corollary_belief_mechanisms}

\restatecorollary{corollary_belief_mechanisms}
\begin{corollary}
	\corollaryBeliefMech
\end{corollary}

\begin{proof}
By Lemma~\ref{lemma_best_no_effort_strategy}, for any spot checking peer prediction mechanism, the $g^l$ strategy is to always report the low-quality signal $s^l$.  

To verify that the conditions of Theorem~\ref{theorem_suff_condition_pareto} are satisfied, it suffices to verify that when $p = 0$, the $s^l$ equilibrium of the peer prediction mechanism exists and Pareto dominates the truthful equilibrium.  We verify these two conditions for all of the listed peer prediction mechanisms below.

Let $b_s$ denote a belief report which predicts that signal $s$ is observed with probability $1$, i.e. $\pr(s) = 1$ and $\pr(s') = 0, \forall s' \in Q, s' \ne s$.
Let the $s^l$ equilibrium denote the equilibrium where every agent's signal report is $s^l$ and belief report is $b_{s^l}$. 

For mathematical convenience, we assume that the scoring rule is \emph{symmetric} \citep{gneiting2007strictly}.  That it, the reward for reporting a signal that is predicted with probability~1 is the same regardless of the signal's identity:
\[R(b_{s}, s) = R(b_{s'}, s'), \forall s \ne s'.\]
This is a very mild condition that is satisfied by all standard scoring rules that compute rewards based purely on the predicted probabilities and the outcome, including the quadratic scoring rule and the log scoring rule.

For symmetric scoring rules, when $p=0$, an agent's expected score is maximized by predicting $b_{s}$ when $s$ is observed for any signal $s \in Q$.  

\paragraph{Binary Robust BTS~\citep{witkowski2012robust,witkowski2013learning}} 

When $c^E = 0$ and $p = 0$, the $s^l$ equilibrium exists. 
(If all other agents report $s^l$ and $b_{s^l}$, then the best belief report for agent $i$ is $b_{s^l}$.  Moreover the best signal report for agent $i$ is $s^l$ which leads to a shadowed belief report of $b_{s^l}$.)

When $c^E = 0$ and $p = 0$, at the $s^l$ equilibrium, an agent's expected utility is $R(b_{s^l}, s^l) + R(s_{s^l}, s^l)$.
This is the maximum possible expected utility that an agent can achieve because the proper scoring rule $R$ is symmetric.  Therefore, it must be greater than or equal to the agent's expected utility at the truthful equilibrium when $c^E = 0$ and $p = 0$.

\paragraph{Multi-valued Robust BTS~\citep{radanovic2013robust}}

When $c^E = 0$ and $p = 0$, the $s^l$ equilibrium exists. 
(If all other agents report $s^l$ and $b_{s^l}$, then the best belief report for agent $i$ is $b_{s^l}$.  Moreover, the best signal report for agent $i$ is $s^l$ which maximizes the probability of his signal report agreeing with other agents' signal reports.)

\smallskip
When $c^E = 0$ and $p = 0$, at the $s^l$ equilibrium, an agent's expected utility is 
\begin{align}
\sum_{s^l} \pr(s^l) \pr(s^l | s^l) + R(b_{s^l},s^l) \notag
= \sum_{s^l} \pr(s^l) + R(b_{s^l}, s^l) \notag
= 1 + R(b_{s^l},s^l), \notag
\end{align}
where the first equality was due to the fact that the low-quality signal $s^l$ is noiseless ($\pr(s^l | s^l) = 1$).

When $c^E = 0$ and $p = 0$, at the truthful equilibrium, an agent's expected utility is 
\begin{align}
&\sum_{s^h \in Q} \pr(s^h) \pr(s^h|s^h) \frac{1}{\pr(s^h|s^h)} + \E[R(\pr(r_j|s^h), r_j)] \notag\\
&= \sum_{s^h \in Q} \pr(s^h) + \E[R(\pr(r_j|s^h), r_j)] \notag
= 1 + \E[R(\pr(r_j|s^h), r_j)] \notag
\le 1 + R(b_{s^l},s^l), \notag
\end{align}
where the inequality was due to the fact that the proper scoring rule $R$ is symmetric.
Thus, the $s^l$ equilibrium Pareto dominates the truthful equilibrium when $c^E = 0$ and $p = 0$.
The conditions of Theorem~\ref{theorem_suff_condition_pareto} are therefore satisfied, and hence $p_{\textup{Pareto}} \ge p_{\textup{ds}}$ for all settings with positive effort cost $c^E \ge 0$.

\paragraph{Divergence-Based BTS~\citep{radanovic2014incentives}} 

When $c^E = 0$ and $p = 0$, the $s^l$ equilibrium exists.
(If all other agents report $s^l$ and $b_{s^l}$, then the best belief report for agent $i$ is $b_{s^l}$.  Moreover, the best signal report for agent $i$ is $s^l$, which means that the penalty is $0$ because the agent's signal reports agree and their belief reports also agree.)

\smallskip
When $c^E = 0$ and $p = 0$, at the $s^l$ equilibrium, an agent's expected utility is 
\begin{align}
- \mathds{1}_{s^l = s^l || D(b_{s^l}, b_{s^l}) > \theta} + R(b_{s^l},s^l) \notag
= R(b_{s^l},s^l).
\end{align}

At the truthful equilibrium, an agent's expected utility is 
\begin{align}
&- \mathds{1}_{s_{i'j}^h = s_{i'j}^h || D(\pr(r|s_{ij}^h), \pr(r|s_{i'j}^h)) > \theta} + R(\pr(r|s^h), s^h)
< R(\pr(r|s^h), s^h) \notag
< R(b_{s^l}, s^l), \notag
\end{align}
where the first inequality was due to the fact that the high-quality signal $s^l$ is noisy.  That is, for every realization $s^h$ of the high-quality signal, $\pr(s^h|s^h) \le 1$ and there exists one realization $s^h$ of the high-quality signal such that $\pr(s^h|s^h) < 1$.  The second inequality was due to the fact that the proper scoring rule $R$ is symmetric.
Thus, the $s^l$ equilibrium Pareto dominates the truthful equilibrium when $c^E = 0$ and $p = 0$.
The conditions of Theorem~\ref{theorem_suff_condition_pareto} are therefore satisfied, and hence $p_{\textup{Pareto}} \ge p_{\textup{ds}}$ for all settings with positive effort cost $c^E \ge 0$.

\paragraph{\citet{riley2014minimum}} 

When $c^E = 0$ and $p = 0$, the $s^l$ equilibrium exists.  
(When all other agents always report $s^l$, for agent $i$, $\delta_i = 0$ because for any signal other than $s^l$, the number of other agents who reported the signal is $0$.  Thus, agent $i$'s reward is $R(b_i, s^l)$.  Since agent $i$'s signal report does not affect his reward, reporting $s^l$ is as good as reporting any other value.  Moreover, since all other agents report $s^l$, the best belief report for agent $i$ is to report $b_{s^l}$.)  

When $c^E = 0$ and $p = 0$, at the $s^l$ equilibrium, $\delta_i = 0$ because for any signal other than $s^l$, the number of other agents who reported the signal is $0$.  Thus, an agent's expected utility is $R(b_{s^l}, s^l)$.
By the definition of the mechanism, an agent's reward is at most $R(b_i, r_{-i})$, which is less than or equal to $R(b_{s^l}, s^l)$ because $R$ is a symmetric proper scoring rule.  Therefore, an agent achieves the maximum expected utility at the $s^l$ equilibrium, which is greater than or equal to the agent's expected utility at the truthful equilibrium when $c^E = 0$ and $p = 0$.
\end{proof}

\end{document}